\documentclass[conference]{IEEEtran}
\IEEEoverridecommandlockouts
\usepackage{cite}
\def\BibTeX{{\rm B\kern-.05em{\sc i\kern-.025em b}\kern-.08em
    T\kern-.1667em\lower.7ex\hbox{E}\kern-.125emX}}

\usepackage{graphicx}

\usepackage{amsmath}
\usepackage{amsthm}
\usepackage{amssymb}
\usepackage{enumerate}
\usepackage{tikz}
\usepackage{url}
\usetikzlibrary{automata,positioning}

\newcommand{\XX}{\mathcal{X}}
\newcommand{\YY}{\mathcal{Y}}
\newcommand{\LL}{\mathcal{L}}
\newcommand{\WW}{\mathcal{W}}
\newcommand{\DD}{\mathcal{D}}
\newcommand{\dist}{\mathbb{D}}
\newcommand{\TT}{\mathcal{T}}
\newcommand{\NN}{\mathbb{N}}
\newcommand{\FF}{\mathcal{F}}
\newcommand{\PP}{\mathcal{P}}
\newcommand{\Ss}{\mathcal{S}}
\newcommand{\A}{\mathcal{A}}
\newcommand{\C}{\mathcal{C}}

\newcommand{\ZZ}{\mathcal{Z}}
\newcommand{\E}{\mathbb{E}}

\DeclareMathOperator{\fst}{fst}
\DeclareMathOperator{\snd}{snd}

\newtheorem{theorem}{Theorem}
\newtheorem{lemma}[theorem]{Lemma}
\newtheorem{definition}[theorem]{Definition}
\newtheorem{proposition}[theorem]{Proposition}
\newtheorem{corollary}[theorem]{Corollary}
\newtheorem{problem}[theorem]{Problem}

\begin{document}

\title{Quantifying information flow in interactive systems\\
\thanks{The author is supported by FNR under grant number 11689058 (Q-CoDe).}}

\author{\IEEEauthorblockN{David Mestel}
\IEEEauthorblockA{\textit{University of Luxembourg} \\
david.mestel@uni.lu}
}

\maketitle

\begin{abstract}
We consider the problem of quantifying information flow in interactive systems,
modelled as finite-state transducers in the style of Goguen and Meseguer. Our
main result is that if the system is deterministic then the information flow is
either logarithmic or linear, and there is a polynomial-time algorithm to
distinguish the two cases and compute the rate of logarithmic flow. To achieve
this we first extend the theory of information leakage through channels to the
case of interactive systems, and establish a number of results which greatly
simplify computation. We then show that for deterministic systems the
information flow corresponds to the growth rate of antichains inside a certain
regular language, a property called the width of the language. In a companion
work we have shown that there is a dichotomy between polynomial and exponential
antichain growth, and a polynomial time algorithm to distinguish the two cases
and to compute the order of polynomial growth. We observe that these two cases
correspond to logarithmic and linear information flow respectively. Finally, we
formulate several attractive open problems, covering the cases of probabilistic
systems, systems with more than two users and nondeterministic systems where the
nondeterminism is assumed to be innocent rather than demonic.
\end{abstract}

\begin{IEEEkeywords}
Quantified information flow, automata theory
\end{IEEEkeywords}

\section{Introduction}
The notion of `noninterference' was introduced by Goguen and Meseguer in
\cite{goguen1982security}.  It has long been recognised, however, that this
condition---that no information can reach Bob about the actions of Alice---may
in some circumstances be too strong.  The field of quantitative information flow
therefore aims to compute the \emph{amount} of information that can reach Bob
about Alice's actions.

The contributions of this work are in two main parts.  In the first part we 
extend the theory of information flow through channels developed by Smith, 
Palamidessi and many others to the case of interactive systems.  In addition to 
basic definitions, we establish a number of results which greatly simplify 
computation.  In particular, we show that it suffices to consider probability 
distributions over deterministic strategies for the two parties and that one of 
them may be assumed to adopt a pure deterministic strategy.  We also show that 
if the system itself is deterministic then there is a possiblistic 
characterisation of the information flow which avoids quantifying over 
probability distributions altogether; this will be essential for the work of 
the second part.

In the second part we study determinstic interactive systems modelled as
finite-state transducers in the style of Goguen and Meseguer.  We define the 
information-flow capacity of such systems, before addressing the formidable 
technical problem of computing it.  The key idea is to show that this 
can be reduced to a certain combinatorial problem on
partially ordered sets. This problem is solved in a companion work 
\cite{mestel2018widths}, with the consequence that we are able to show (Theorem
\ref{thm:entcap}) that for such systems there is a dichotomy between logarithmic
and linear information flow, and a polynomial-time algorithm to distinguish the
two cases. These two cases are naturally interpreted as `safe' and `dangerous'
respectively, so we have shown that it is possible to distinguish genuinely
dangerous information flow.  We thereby accomplish a goal proposed by Ryan, 
McLean, Millen and Gligor at CSFW'01 in \cite{ryan2001noninterference}.

\subsection*{Overview}
In Section \ref{sec:channels} we first recall some relevant theory on the
information-flow capacity of channels, and improve a result of Alvim,
Chatzikokolakis, McIver, Morgan, Palamidessi and Smith giving an upper bound on
the `Dalenius leakage' of a channel to an exact formula (Theorem
\ref{thm:delan}).  We then consider interactive channels, where both parties may
be required to make choices.  We define leakage and information-flow capacity in
this setting, and show that Bob's strategy may be assumed to be deterministic
(Corollary \ref{cor:bobdet}).  We show (Theorem \ref{thm:aliceposs})
that in the case of deterministic channels we may take a possbilistic view of
Alice's actions, which we will find simplifies calculation considerably.  
Finally we show (Theorem \ref{thm:probstrat}) that for systems which may involve 
multiple rounds of interaction it suffices to consider probability distributions 
over deterministic, rather than probabilistic, strategies.

In Section \ref{sec:detint} we model deterministic interactive systems as
finite-state transducers, and define their information-flow capacity. We then
show how to reduce the problem of computing this to a problem involving only
nondeterministic finite automata, and then to the combinatorial problem of
computing the `width' of the languages generated by the relevant automata. We
observe that this problem is solved in \cite{mestel2018widths}, and consequently
conclude (Theorem \ref{thm:entcap}) that there is a dichotomy between
logarithmic and linear information flow, and a polynomial-time algorithm to
distinguish the two cases.  The structure of the sequence of reductions leading
to this theorem is summarised in Figure \ref{fig:chapstruct}, and we illustrate 
the theory by applying it to a simple scheduler.

In Section \ref{sec:future} we discuss some generalisations of the systems
studied in main part of this work: namely nondeterministic systems, systems with
more than two agents (which we observe encompasses the case of nondeterministic
systems), and probabilistic systems.  For the latter two we define the
information-flow capacity and formulate the open problems of computing it.  
Finally in Section \ref{sec:relatedwork} we discuss related work and in Section 
\ref{sec:conclusion} we conclude.

\section{Information-theoretic preliminaries}\label{sec:channels}

\subsection{Leakage through channels}\label{sec:simpchan}

We consider first the case of leakage through a channel from a space $\XX$ of
inputs to a space $\YY$ of outputs, corresponding to a situation in which the
attacker is purely passive: Alice selects an input according to a known prior
distribution $p_X$ and Bob (the attacker) receives an output according to the
conditional distribution $p_{Y|X}$, which specifies the channel.  How much
information should we say that Bob has received?

The first work on quantified information flow adopted the classical
information-theoretic notion of \emph{mutual information} introduced by Shannon
in the 1940s~\cite{shannon1948communication}.  However, Smith observed
in~\cite{smith2009foundations} the problems with this consensus definition.  The
essential problem is that mutual information represents in some sense the
average number of bits of information leaked by the system.  This is appropriate
for the noisy coding theorem, where we are interested in the limit of many uses
of the channel, but not for the case of information leakage where we assume that
the adversary receives only one output (or a small number of outputs).

This means that, in the example used by Smith, a system which leaks the whole
secret 1/8 of the time is seen as largely secure (because $H(X|Y)=\tfrac{7}{8}H(X)$),
although it allows (for instance) a cryptographic key to be guessed 1/8 of the
time.  Smith addresses this by adopting the \emph{min-entropy leakage}, 
defined\footnote{Smith and subsequent authors generally define leakage only for 
random variables whose images are finite sets.  However, their definitions are 
straightforwardly generalised to arbitrary discrete random variables by replacing 
$\max$ with $\sup$ where appropriate.  Except where noted, the proofs of all 
quoted results remain valid after the same modification.}
as the expected value of the increase in the probability of guessing the input
upon observing the output $y$:
\[\LL_\infty(X,Y) = \log \mathbb{E}_{y\sim Y} \frac{\sup_{x\in\XX}
p_{X|Y}(x|y)}{\sup_{x\in\XX} p_X(x)}.\]

Given a channel $\C$ specified by a matrix of conditional probabilities
$p_{Y|X}$, we may be interested in its \emph{capacity}, which is the maximum
value of the leakage over all possible priors $p_X$:
\[\LL_\infty(\C) = \sup_{(X,Y)\sim \C} \LL_\infty(X,Y),\]
where the notation $(X,Y)\sim \C$ means that $X$ and $Y$ are random variables
compatible with $\C$; that is, that the conditional probabilities $p_{Y|X}$
(where defined) correspond to the matrix defining $\C$.

In \cite{alvim2012measuring}, Alvim, Chatzikolakis, Palamidessi and Smith
generalise this definition to the notion of \emph{$g$-leakage}, in which Bob
makes a guess drawn from a set $\WW$, and receives a payoff according to the
function $g:\WW\times\XX \rightarrow [0,1]$.  The leakage with respect to $g$ is
then
\[\LL_g(X,Y) = \log \mathbb{E}_{y\sim Y} \frac{\sup_{w\in \WW} \sum_{x\in \XX}
p_{X|Y}(x|y)g(w,x)}{\sup_{w\in \WW} \sum_{x\in \XX} p_{X}(x)g(w,x)}.\]

Once again, we can define the capacity of a channel $\C$:
\[\LL_g(\C) = \sup_{(X,Y)\sim \C} \LL_g(X,Y).\]

In Theorem 5.1 of \cite{alvim2012measuring}, the authors prove the so-called
`miracle' theorem, which states that for any channel $\C$ and any gain function
$g$ we have that the $g$-capacity is at most the min-entropy capacity:
\[\LL_g(\C) \leq \LL_\infty(\C).\]

However, it may be the case that the secret which Bob is trying to guess is not
Alice's input but some other secret value (a cryptographic key, say) which is
related to $x$ in some known but unspecified way.  We may be interested in
bounding the possible gain for Bob for any possible secret and any
(probabilistic) relationship to the choice of $x$; this is sometimes known as
the `Dalenius leakage', after a desideratum attributed to T. Dalenius by Dwork
in \cite{dwork2006differential}.  We may therefore define
\[\LL_D(X,Y) = \sup_{Z \in \mathcal{D}} \LL_\infty(Z,Y),\]
where $\DD$ is the collection of random variables $Z$ such that $Z\rightarrow X
\rightarrow Y$ forms a Markov chain (that is, $p_{X,Y,Z}(x,y,z) =
p_Z(z)p_{X|Z}(x|z) p_{Y|X}(y|x)$).

In \cite{alvim2014additive}, Alvim, Chatzikokolakis, McIver, Morgan, Palamidessi
and Smith give an upper bound for the Dalenius leakage: they show in Corollary
23 that for any Markov chain $Z\rightarrow X\rightarrow Y$ we have that 
\[\sup_g \LL_g(Z,Y) \leq \sup_g \LL_g(Y,X),\]
where the suprema are taken over gain functions $g$.  Hence in particular we
have that $\LL_\infty(Z,Y) \leq \sup_g \LL_g(Y,X)$.  But $\LL_g(Y,X) \leq
\LL_g(\C)$, where $\C$ is any channel such that $(X,Y)\sim \C$, and by the
miracle theorem we have that $\LL_g(\C) \leq \LL_\infty(\C)$, and hence we have
that
\[\LL_D(X,Y) \leq \LL_\infty(\C).\]

We are able to improve this to a precise formula for the Dalenius leakage
between two random variables.

\begin{theorem}\label{thm:delan}
Let $X,Y$ be any discrete random variables.  Then 
\begin{align*}\LL_D(X,Y) &= \log \mathbb{E}_{y\sim Y} \sup_{x\in \XX} \frac{p_{Y|X}(y|x)}
        {p_Y(y)}\\
        &= \log \sum_{y\in \YY^+} \sup_{x \in \XX} p_{Y|X} (y|x),\end{align*}
where $\YY^+\subseteq \YY$ is the set of $y\in \YY$ such that $p_Y(y) > 0$.
\end{theorem}
\begin{proof}
We may assume without loss of generality that $p_Y(y) > 0$ for all $y\in Y$
(otherwise redefine $\YY$ to be the set of values on which $p_Y$ is supported).

For the upper bound, we recall that Braun, Chatzikokolakis and Palamidessi
observe in Proposition 5.11 of \cite{braun2009quantitative} that there is a
simple formula for the min-entropy capacity of a channel $\C$ defined by matrix
$p_{Y|X}$:
\begin{equation}\label{eq:minentcap} \LL_\infty(\C) = \log \sum_{y\in \YY}
\sup_{x\in \XX} p_{Y|X}(y|x).\end{equation}

This is proved in \cite{braun2009quantitative} for random variables with finite 
image.  For general discrete random variables, the upper bound on 
$\LL_\infty(\C)$ is obtained by replacing $\max$ with $\sup$ as appropriate, but 
the lower bound requires a little more care since it is given by considering the 
uniform distribution on $\XX$.  However, the lower bound can be recovered for 
infinite $\XX$ by considering the uniform distribution on the first $k$ elements 
of $\XX$ and taking the limit as $k\rightarrow \infty$.

The upper bound on $\LL_D(X,Y)$ is immediate from \eqref{eq:minentcap} by taking 
$\C$ to be any channel such that $(X,Y)\sim \C$.

For the lower bound, suppose that $\XX=\{x_1,x_2,x_3,\ldots\}$, and define the
function $f:[0,1)\rightarrow \XX$ by $f(\xi) = x_k$ if 
\[\sum_{i=1}^{k-1} p_X(x_i) < \xi \leq \sum_{i=1}^k p_X(x_i).\]
Note that for each $x\in \XX$ we have that $p_X(x) = \mu(f^{-1}(x))$, where
$\mu$ is the Borel measure.

For each positive integer $n$, let $\ZZ_n = \{0,1,2,\ldots,2^n-1\}$ and let
$Z_n$ be a random variable taking values in $\ZZ_n$, with 
\[p_{X,Z_n} (x,z) = \mu\left(f^{-1}(x)\cap \left[\frac{z}{2^n},\frac{z+1}{2^n}
\right]\right).\]
Note that by the previous observation we have $\sum_{z\in \ZZ_n} p_{X,Z_n}(x,z)
= \mu(f^{-1}(x)) = p_X(x)$ as required.  Note also that we have that $p_Z(z) =
\sum_{x\in \XX} p_{X,Z_n}(x,z) = 2^{-n}$.

Now we have
\begin{align*}\LL_\infty(Z_n,Y) &= \log\sum_{y\in \YY} p_Y(y) \frac{\max_z
p_{Z_n|Y}(z|y)}{\max_z p_{Z_n}(z)} \\
&= \log\sum_{y\in \YY} p_Y(y) 2^n \max_z p_{Z_n|Y} (z|y).\end{align*}
We claim that 
\begin{equation}\label{eq:limzy}
\lim_{n\rightarrow \infty} p_Y(y) 2^n \max_{z} p_{Z_n|Y}(z|y)
\geq \sup_x p_{Y|X} (y|x)\end{equation}
for all $y\in \YY$.

Indeed, by Bayes' theorem we have 
\begin{align*}p_Y(y)p_{Z_n|Y}(z|y)&=p_{Z_n}(z)p_{Y|Z_n}(y|z) \\
&= 2^{-n} \sum_{x\in \XX} p_{Y|X}(y|x) p_{X|Z_n}(x|z).\end{align*}  
Let $x\in \XX$ be arbitrary.  
For sufficiently large $n$ we have that $p_{X|Z_n}(x|z) = 1$
for some $z\in \ZZ_n$, and hence for this $z$ we have that 
$p_Y(y)2^np_{Z_n|Y}(z|y) \geq p_{Y|X}(y|x)$, proving the claim.  
Summing (\ref{eq:limzy}) over all $y\in \YY$ and rearranging gives
\[\lim_{n\rightarrow \infty} \LL_\infty(Z_n,Y) \geq \log \sum_{y\in \YY} \max_x
p_{Y|X} (y|x),\]
as required.
\end{proof}

\newpage

\subsection{Interactive channels}\label{sec:intchan}

More generally, we will be interested in \emph{interactive channels}, where an
input is chosen by both Alice and Bob, and the system then produces an output to
Bob.  This means that the space $\XX$ is of the form $\XX_A\times \XX_B$, where
the spaces $\XX_A$ and $\XX_B$ are the spaces of inputs for Alice and Bob
respectively, and the interactive channel $\C$ is defined by the matrix of
conditional probabilites $p_{Y|X_A,X_B}$.  

Note that if the system involves a
sequence of outputs and actions by Alice and Bob then the `inputs' $x_A$ and
$x_B$ will in fact represent \emph{strategies} for Alice and Bob, determining
their actions on the basis of the outputs they have seen so far (in general 
these may be probabilistic, but we will see in Section \ref{sec:probstrat} that 
in fact it is sufficient to consider only deterministic strategies).  

We will write $((X_A,X_B),Y)\sim \C$ to mean that the random variables $X_A,X_B$
and $Y$ are consistent with the channel $\C$: that is, that $X_A$ and $X_B$ are
independent and the matrix $p_{Y|X_A,X_B}$ corresponds with the matrix defining
$\C$.

We can once again define the min-entropy leakage as the expected increase in
Bob's probability of guessing the value of the input based on having seen the
output:
\begin{align*}
&\LL_\infty ((X_A,X_B),Y) \\
&\qquad= \log \mathop{\E}_{x_B\sim X_B, y\sim Y} \frac{\sup_{x_A\in
\XX_A} p_{X_A|X_B,Y}(x_A|x_B,y)} {\sup_{x_A\in \XX_A} p_{X_A}(x_A)}\\
&\qquad= \log \E_{x_B\sim X_B} 2^{\LL_\infty(X_A,Y|X_B=x_B)}.
\end{align*}
Again the capacity of the channel is defined as the maximum leakage over all
possible priors $p_{X_A}$ and $p_{X_B}$.
\[\LL_\infty(\C) = \sup_{((X_A,X_B),Y)\sim \C} \LL_\infty((X_A,X_B),Y).\]

It appears at first glance that calculating $\LL_\infty(\C)$ may in general be
highly intractible: we have to quantify over mixed strategies (that is over
probability distributions on strategies) for Alice and Bob.  However, it turns
out that we may assume without loss of generality that Bob chooses a pure
strategy.\footnote{Note that this means a pure strategy \emph{over the set 
$\XX_B$}, which in an interactive system may contain probabilistic strategies 
(although we will see in Theorem \ref{thm:probstrat} that these may be ignored 
without loss of generality).} Indeed, this holds not only for the choices we have made but for all
reasonable such choices.  

Specifically, we chose a leakage measure, namely $\LL_\infty$, and a method of
averaging the leakage over different values of $x_B$, namely taking $\log
\E_{x_B} 2^\LL$.  The following proposition shows that we may assume a pure
strategy for Bob for any choice of leakage measure, and any method of averaging
which is `reasonable' in the sense that if the distribution of leakage is
constant with value $x$ then the value is $x$, and also that the value of a
weighted sum of leakage distributions cannot be more than the maximum value of
the distributions making up the sum (this last property is known as `quasiconvexity').

\begin{proposition}\label{prop:genpure}
Let $\LL:\mathbb{D}(\XX_A\times \YY) \rightarrow \mathbb{R}$ (the `leakage
function') be any function and let $\phi:\mathbb{D}(\mathbb{R}) \rightarrow
\mathbb{R}$ (the `averaging function') be any function such that if $X\in
\mathbb{D}(\mathbb{R})$ is constant $x$ then $\phi(X) = x$ and for any
$X_1,X_2,\ldots \in \mathbb{D}(\mathbb{R})$ and any $\rho_1,\rho_2,\ldots$ with
$\sum_i \rho_i =  1$ we have
\begin{equation}\label{eq:phinice}\phi\left(\sum_i \rho_i X_i\right) \leq \sup_i
\phi(X_i).\end{equation}
Let
\[\LL_\phi(\C) = \sup_{((X_A,X_B),Y)\sim \C} \phi(\LL(X_A,Y|X_B)).\]
Then we have
\[\LL_\phi(\C) = \sup_{x_B\in \XX_B} \sup_{(X_A,x_B,Y) \sim \C}
\LL(X_A,Y),\]
where the notation $(X_A,x_B,Y)$ means the distribution with $p_{X_B}(x_B) = 1$,
and in the above $\mathbb{D}(\XX)$ means the space of probability distributions
over the set $\XX$.
\end{proposition}
\begin{proof}
Suppose that $(X_A,X_B,Y)\sim \C$.  We have
\[\phi(\LL(X_A,Y|X_B)) = \phi\left(\sum_{x_B} p_{X_B}(x_B) \LL(X_A,Y|X_B=x_B)
\right).\]
Hence for any $\epsilon > 0$, by (\ref{eq:phinice}) there exists some $x_B$ such
that
\begin{align*}\phi(\LL(X_A,Y|X_B=x_B)) &= \LL(X_A,Y|X_B=x_B) \\
&\geq \phi(\LL(X_A,Y|X_B)) - \epsilon.\end{align*}
Hence we have
\[\sup_{x_B\in\XX_B} \LL(X_A,Y|X_B=x_B) = \phi(\LL(X_A,Y|X_B)),\]
establishing the result.
\end{proof}

The min-entropy capacity is a special case of this result, with $\LL=\LL_\infty$
and $\phi(X) = \log\E_{x\sim X} 2^x$.

\begin{corollary}\label{cor:bobdet}
Let $\C$ be an interactive channel.  Then we have
\[\LL_\infty(\C) = \sup_{x_B\in\XX_B} \LL_\infty(\C|X_B=x_B).\]
\end{corollary}

\subsection{Deterministic channels}

For the channels we have considered above, once the inputs from Alice and Bob
are fixed we obtain a probability distribution on outputs.  However, for some
systems it may be that the output is not probabilistic, but is determined by the
values of the inputs; we will call such a channel deterministic.  More
concretely, an interactive channel $\C$ defined by the matrix $p_{Y|X_A,X_B}$ is
\emph{deterministic} if for all $x_A,x_B,y$ we have
\[p_{Y|X_A,X_B}(y|x_A,x_B) \in \{0,1\}.\]

If $\C$ is deterministic then the computation of $\LL_\infty(\C)$ simplifies
considerably, because it turns out that we can take a purely possibilistic view
of Alice's actions and avoid any quantification over probability distributions.

\begin{theorem}\label{thm:aliceposs}
Let $\C$ be a deterministic interactive channel.  Then
\begin{multline*}\LL_\infty(\C) = \sup_{x_B\in \XX_B} \log \left| \left\{y\in\YY \middle\vert \exists
x_A\in \XX_A:\right.\right.\\
\left.\left. p_{Y|X_A,X_B}(y|x_A,x_B)=1\right\}\right|.\end{multline*}
\end{theorem}
\begin{proof}
By Corollary \ref{cor:bobdet}, it suffices to prove that 
\begin{multline*}\LL_\infty(\C|X_B=x_B) = \log \left| \left\{y\in\YY \middle\vert \exists
x_A\in \XX_A:\right.\right.\\
\left.\left. p_{Y|X_A,X_B}(y|x_A,x_B)=1\right\}\right|.\end{multline*}
By the formula for $\LL_\infty(\C)$ from
\cite{braun2009quantitative} (recalled as (\ref{eq:minentcap}) in the proof of
Theorem \ref{thm:delan}) we have
\begin{multline*}\LL_\infty(\C|X_B=x_B) = \log \sum_{y\in \YY} \max_{x_A\in \XX_A}
p_{Y|X_A,X_B} (y|x_A,x_B) \\
= \log \left| \left\{y\in\YY \middle\vert \exists
x_A\in \XX_A: p_{Y|X_A,X_B}(y|x_A,x_B)=1\right\}\right|,
\end{multline*}
since $\C$ is deterministic and so $p_{Y|X_A,X_B}(y|x_A,x_B) \in \{0,1\}$.
\end{proof}

Theorem \ref{thm:aliceposs} essentially says that it suffices to count the
maximum number of outputs that can be seen by Bob, consistently with his choice
of strategy.  The corresponding result for non-interactive channels is Theorem 1
of \cite{smith2009foundations}.

\subsection{Probabilistic vs deterministic strategies}\label{sec:probstrat}

We observed in Section \ref{sec:intchan} that the `channel' paradigm is able to
model systems involving many rounds of interaction, because we can take Alice
and Bob's inputs to be strategies, determining the actions they will take at
each step of the interaction. At each step, Alice (respectively Bob) will have
observed a trace of the interaction thus far drawn from a set $T$, and must
select an action drawn from a set $\Sigma$. To specify a randomised strategy for
Alice or Bob, we must therefore specify for each $t\in T$ a probability
distribution over $\Sigma$, so the set of strategies is the set of maps
$T\rightarrow \dist \Sigma$.

In this section we will show that in fact it suffices to consider only 
deterministic strategies for Alice and Bob.  The intuition behind this is fairly 
straightforward: given a probabilistic strategy, we could imagine that any 
necessary coins are tossed before the execution begins, which gives a probability 
distribution over deterministic strategies.  This changes nothing except that it 
allows Bob to see how the random choices made by his strategy were resolved, but
this only gives him more information and so does not affect the information flow
capacity.
To avoid technical measurability issues we 
will assume that the sets $T$ and $\Sigma$ are finite.

\begin{definition}
Let $T$ be a finite set of \emph{traces} and $\Sigma$ a finite set of 
\emph{actions}.  A \emph{strategy} over $T$ and $\Sigma$ is a function $f:T\rightarrow 
\dist(\Sigma)$.  The set of strategies over $T$ and $\Sigma$ is denoted 
$\Ss_{T,\Sigma}$.

A strategy $f\in \Ss_{T,\Sigma}$ is \emph{determinsitic} if we have 
\[f(t)(x) \in \{0,1\}\]
for all $t\in T$ and $x\in \Sigma$.  We write $\DD_{T,\Sigma} \subset 
\Ss_{T,\Sigma}$ for the set of deterministic strategies over $T$ and $\Sigma$.
\end{definition}

In the execution itself, these strategies will be executed and particular 
actions chosen.  The output $y\in \YY$ displayed to Bob is then a function 
(which may be probabilistic) of the choices that were made; the system is defined by this 
function, which is a map from pairs of functions $T \rightarrow \Sigma$ (the 
choices made by Alice and Bob respectively) to distributions over $\YY$.  
Note that it may be that in some 
executions not all traces are actually presented to Alice and Bob for decision; 
this can be represented by the choices made in response to those traces being 
ignored, so no generality is lost by considering total functions $T\rightarrow 
\Sigma$ (similarly the trace-sets relevant to Alice and Bob may be distinct, but 
this can be represented by ignoring the choices made by Alice on Bob's traces 
and vice versa).

We write $\Sigma^T$ for the set of functions $T\rightarrow \Sigma$; the 
probability that a particular function is realised by a particular strategy can 
be computed by multiplying the probabilities for each decision (note that 
nothing is lost by assuming independence: if Alice and Bob are supposed to know 
about previous choices they have made then this can be encoded in the traces).

\begin{definition}\label{def:realise}
Let $f\in \Ss_{T,\Sigma}$ be any strategy and $g\in \Sigma^T$.  The probability 
that $f$ \emph{realises} $g$, written $f(g)$, is given by
\[f(g) = \prod_{t\in T} f(t)(g(t)).\]
\end{definition}

\begin{definition}\label{def:output}
Let $\phi:\Sigma^T\times \Sigma^T \rightarrow \dist \YY$ be any map, and let 
$\XX_A$ and $\XX_B$ be any subsets of $\Ss_{T,\Sigma}$.  The \emph{interactive 
channel determined by $\phi, \XX_A$ and $\XX_B$}, denoted $\C_{\phi,\XX_A,\XX_B}$,
is determined by the matrix of conditional probabilities
\[p_{Y|X_A,X_B}(y|f_A,f_B) = \sum_{g_A,g_B\in \Sigma^T} 
f_A(g_A)f_B(g_B) \phi(g_A,g_B)(y).\]
\end{definition}

We observe that if the function $\phi$ defining the system is deterministic, and 
if Alice and Bob use only deterministic strategies, then the channel produced is 
a deterministic interactive channel in the sense of the previous section, such 
that Theorem \ref{thm:aliceposs} applies to it.

\begin{proposition}
Suppose that $\phi(g,g')(y)\in \{0,1\}$ for every $g,g'\in \Sigma^T$ and $y\in 
\YY$.  Then $\C_{\phi,\DD_{T,\Sigma},\DD_{T,\Sigma}}$ is a deterministic 
interactive channel.
\end{proposition}
\begin{proof}
If $f_A,f_B\in \DD_{T,\Sigma}$ then $f_A(g),f_B(g) \in \{0,1\}$ for all $g\in 
\Sigma^T$.  Hence if $\phi(g_A,g_B,y) \in \{0,1\}$ for all $g_A,g_B,y$ then 
we have $p_{Y|X_A,X_B}(y|f_A,f_B) \in \{0,1\}$ for all $f_A,f_B,y$, as required.
\end{proof}

The main theorem of this section is that in fact it suffices to conisder only 
deterministic strategies for Alice and Bob.

\begin{theorem}\label{thm:probstrat}
Let $\Sigma$ and $T$ be any finite sets, $\YY$ any set and $\phi:\Sigma^T \times
\Sigma^T \rightarrow \dist \YY$ be any map.  Then we have
\[\LL_\infty\left( \C_{\phi,\Ss_{T,\Sigma},\Ss_{T,\Sigma}} \right) 
= \LL_\infty\left( \C_{\phi,\DD_{T,\Sigma},\DD_{T,\Sigma}} \right).\]
\end{theorem}
\begin{proof}
The lower bound is immediate: since whenever $((X_A,X_B),Y)\sim 
\C_{\phi,\DD_{T,\Sigma},\DD_{T,\Sigma}}$ then also $((X_A,X_B),Y) \sim 
\C_{\phi,\Ss_{T,\Sigma},\Ss_{T,\Sigma}}$, we must have (writing $\C_{\Ss}$ and 
$\C_{\DD}$ resectively for the two channels in the statement of the theorem)
\begin{align*}
\LL_\infty(\C_{\Ss}) &= 
\sup_{((X_A,X_B),Y)\sim \C_\Ss} 
\LL_\infty((X_A,X_B),Y) \\
&\geq \sup_{((X_A,X_B),Y)\sim \C_\DD} 
\LL_\infty((X_A,X_B),Y) \\
&= \LL_{\infty}(\C_\DD).
\end{align*}

For the upper bound, let $X_A$ and $X_B$ be any independent $\Ss_{T,\Sigma}$-valued random 
variables.  We will first show that without loss of generality we may assume 
that $X_B$ is supported only on $\DD_{T,\Sigma}$.  By Corollary \ref{cor:bobdet} 
it suffices to show this where $X_B$ is a point distribution,\footnote{Strictly 
speaking Corollary \ref{cor:bobdet} was proved for discrete distributions, 
whereas $\Ss_{T,\Sigma}$ is a continuous subset of $\mathbb{R}^{|T|\cdot 
|\Sigma|}$.  The proof for this case is exactly the same, with sums over $\XX_B$ 
replaced by integrals with respect to the Lebesgue measure.} so say that $X_B$ 
takes the value $f_B\in \Ss_{T,\Sigma}$.

Define the random variable $X_B'$ to be supported only on $\DD_{T,\Sigma}$, and 
for $f \in \DD_{T,\Sigma}$ let
\[p_{X_B'}(f) = f_B(\widetilde{f}),\]
where $\widetilde{f}$ is the function $T\rightarrow \Sigma$ induced by $f$: that 
is, $\widetilde{f}(t)$ is the unique element $x$ of $\Sigma$ such that 
$f(t)(x)=1$.

Note that by Definitions \ref{def:realise} and \ref{def:output} we have that 
$(X_A,X_B)$ and $(X_A,X_B')$ induce the same output distribution $Y$, and so it 
suffices to prove that for each $y\in \YY$ we have
\begin{multline*}\E_{f'_B\sim X_B'} \sup_{f_A\in \Ss_{T,\Sigma}} p_{X_A|X_B',Y} (f_A|f_B',y) 
\geq \\ \sup_{f_A\in \Ss_{T,\Sigma}} p_{X_A|X_B,Y} (f_A|f_B,y).\end{multline*}
Now on the one hand we have
\begin{multline}\label{eq:xbbig}
\E_{f'_B\sim X_B'} \sup_{f_A\in \Ss_{T,\Sigma}} p_{X_A|X_B',Y} (f_A|f_B',y) = \\
\sum_{f_B'\in \DD_{T,\Sigma}} f_B(\widetilde{f_B'}) \sup_{f_A \in \Ss_{T,\Sigma}} 
p_{X_A|X_B',Y} (f_A|f_B',y).\end{multline}
On the other hand we have
\begin{multline}\label{eq:xbsmall}
\sup_{f_A\in \Ss_{T,\Sigma}} p_{X_A|X_B,Y}(f_A|f_B,y) = \\
\sup_{f_A\in \Ss_{T,\Sigma}} \sum_{f_B'\in \DD_{T,\Sigma}} f_B(\widetilde{f_B'}) 
p_{X_A|X'_B,Y} (f_A|f_B',y).\end{multline}
Plainly $\eqref{eq:xbbig} \geq \eqref{eq:xbsmall}$, establishing the result.

We now show that we may also assume that $X_A$ is supported only on 
$\DD_{T,\Sigma}$, and again by Corollary \ref{cor:bobdet} it suffices to show 
this where $X_B$ takes only a single value, say $f_B\in \DD_{T,\Sigma}$.  By 
the min-entropy capacity formula \eqref{eq:minentcap} conditioned on $X_B=f_B$ 
it suffices to show that for every $y\in \YY$ we have
\[\sup_{f_A\in \Ss_{T,\Sigma}} p_{Y|X_A,X_B}(y|f_A,f_B) \leq 
\max_{f_A\in \DD_{T,\Sigma}} p_{Y|X_A,X_B}(y|f_A,f_B).\]
But this is straightforward: indeed, for any $f_A\in \Ss_{T,\Sigma}$ we have
\begin{align*}
p_{Y|X_A,X_B}(y|f_A,f_B) &= \sum_{f_A'\in \DD_{T,\Sigma}} f_A(\widetilde{f_A'}) 
\phi(\widetilde{f_A'},\widetilde{f_B},y) \\
&\leq \max_{f_A'\in \DD_{T,\Sigma}} \phi(\widetilde{f_A'},\widetilde{f_B},y) \\
&= \max_{f_A'\in \DD_{T,\Sigma}} p_{Y|X_A,X_B}(y|f_A',f_B,y),\end{align*}
as required.
\end{proof}

\section{Deterministic interactive systems}\label{sec:detint}

\subsection{Finite-state transducers}

We will model deterministic interactive systems as deterministic finite-state
transducers.  Whereas Goguen and Meseguer in \cite{goguen1982security} modelled 
such systems as `state-observed' transducers, we will consider the more general notion of `action-observed' transducers (see the work of van der Meyden and Zhang in 
\cite{vandermeyden2007comparison} for further discussion of the relationship 
between noninterference properties in these two models; this model is also 
essentially equivalent to the notion of `Input-Output Labelled Transition System'
used by Clark and Hunt in the non-quantitative setting in 
\cite{clark2009noninterference}).

\begin{definition}
A \emph{deterministic finite-state transducer} (DFST) is a 7-tuple $\TT =
(Q,q_0,F,\Sigma,\Gamma,\delta,\sigma)$, where $Q$ is a finite set of
\emph{states}, $q_0\in Q$ is the \emph{initial state}, $F\subseteq Q$ is the set
of \emph{accepting states}, $\delta:Q\times \Sigma\rightarrow Q$ is the
\emph{transition function} and $\sigma:Q\times \Sigma\rightarrow \Gamma\cup
\{\epsilon\}$ is the \emph{output function}.
\end{definition}

A pair $(a_1a_2\ldots a_k, b_1b_2\ldots b_l)\in \Sigma^*\times \Gamma^*$ is
accepted by $\TT$ if there exists a sequence of states $q_1\ldots q_k\in Q^*$
such that $q_k\in F$, for every $0\leq i < k$ we have $q_{i+1} = \delta(q_i,
a_{i+1})$ and $b_1\ldots b_l = \sigma(q_0,a_{1}) \sigma(q_1,a_2) \ldots
\sigma(q_{k-1},a_k)$.  We will write $L(\TT)$ for the subset of
$\Sigma^*\times \Gamma^*$ accepted by $\TT$; such a set is a deterministic
finite-state \emph{transduction}, which we will also abbreviate by DFST.

This definition is not quite convenient for our purposes, because we
assume that the agents are able to observe the passage of time.  Hence even at a
timestep where the machine does nothing, there should be a record in the trace
of the fact that time has passed.  We ensure this by requiring that there
should be an output at each step, and apply the non-standard term `synchronised'
to describe this property (such a transducer is also sometimes called
`letter-to-letter').

\begin{definition}A DFST $\TT = (Q, q_0, F, \Sigma, \Gamma, \delta, \sigma)$ is
\emph{synchronised} if $\sigma(Q,\Sigma)\subseteq \Gamma$ (that is, we do not
have $\sigma(q,a)=\epsilon$ for any $q\in Q$ and $a\in \Sigma$).  In this case
we say that $\TT$ is a synchronised deterministic finite-state transducer
(SDFST).
\end{definition}

Note that this definition almost corresponds with the original definition of a
\emph{Mealy machine} (\cite{mealy1955method}), except that we allow for a set of
final states $F\neq Q$.  It is clear that if $\TT$ is synchronised then
$(a_1\ldots a_k,b_1\ldots b_l)\in \Sigma^*\times \Gamma^*$ is accepted by $\TT$
only if $l=k$.  We shall therefore apply the `zip' operation and view $\TT$ as
accepting elements of $(\Sigma\times \Gamma)^*$.

We are interested in SDFSTs of a special kind, representing the fact that the
system communicates separately with Alice and Bob.  We will consider SDFSTs
whose input and output alphabets $\Sigma$ and $\Gamma$ are of the form
$\Sigma_A\times \Sigma_B$ and $\Gamma_A\times \Gamma_B$ respectively.  The pairs
$(\Sigma_A,\Gamma_A)$ and $(\Sigma_B,\Gamma_B)$ represent the input and output
alphabets used for communication with Alice and Bob respectively.

A simple example of such a transducer is the system which simply relays messages
between the two agents (with $\Sigma_A=\Sigma_B=\{a,b\}$ and $\Gamma_A=\Gamma_B 
=\{a',b'\}$).
This is shown in Figure \ref{fig:relay}.

\begin{figure}[htbp]
\centering
\begin{tikzpicture}[shorten >=1pt,node distance=2cm,on grid,auto]
   \node[state,initial, accepting] (q_0)   {$q_0$};
   \path[->]
    (q_0) edge [loop above] node {$(x,y)|(y',x') \forall x,y\in \{a,b\}$} ();
\end{tikzpicture}
\caption{A relay system.}\label{fig:relay}
\end{figure}
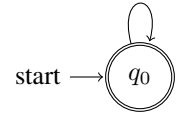

\subsection{Strategies and information flow}

In order to apply the framework of the previous section, we must define the
spaces $\XX_A, \XX_B$ of strategies for Alice and Bob, the space $Y$ of outcomes
visible to Bob, and the matrix $p_{Y|X_A,X_B}$ governing which
outcomes occur.  Since we are considering deterministic specifications, the
matrix $p_{Y|X_A,X_B}$ will be 0-1-valued.

Alice and Bob must each decide on an action based on the trace they have seen
thus far, so a strategy for Alice is a function 
\[x_A:(\Sigma_A\times\Gamma_A)^* \rightarrow \Sigma_A,\]
and similarly a strategy for Bob is a function $x_B:(\Sigma_B\times \Gamma_B)^*
\rightarrow \Sigma_B$.  

Recall that by Theorem \ref{thm:probstrat} it suffices to
consider deterministic strategies for Alice and Bob:  
in the language of Section 
\ref{sec:probstrat}, we have $T=(\Sigma_A\times\Gamma_A)^*\cup (\Sigma_B\times 
\Gamma_B)^*$ and $\Sigma = \Sigma_A\cup \Sigma_B$. We will have that the 
function $\phi(g_A,g_B,y)$ ignores the values of $g_A$ on $(\Sigma_B\times 
\Gamma_B)^*$ and the values of $g_B$ on $(\Sigma_A\times \Gamma_A)^*$, and treats 
all elements of $\Sigma_B$ in the image of $g_A$ as equivalent to some fixed 
$a\in \Sigma_A$ and similarly all elements of $\Sigma_A$ in the image of $g_B$ 
as equivalent to some fixed $b\in \Sigma_B$.  By Theorem \ref{thm:probstrat} it 
suffices to consider deterministic strategies for Alice and Bob and so it is 
more convenient to refer to the sets of deterministic strategies directly as 
$\XX_A$ and $\XX_B$, and to $\phi(x_A,x_B)(y)$ directly as the channel matrix 
$p_{Y|X_A,X_B}(y|x_A,x_B)$.

Given an SDFST $\TT$, and strategies $x_A$ and $x_B$ for Alice and Bob
respectively, what output or outputs can be shown to Bob?  We consider first
the case where $F=Q$, postponing for later the issues that arise when
$F\subsetneq Q$.

\begin{definition}\label{def:cons}We will say that a word
\begin{multline*}w=((a_1,a_1'),(b_1,b_1'))\ldots ((a_k,a_k'),(b_k,b_k')) \\ \in (\Sigma\times
\Gamma)^* = \left((\Sigma_A\times \Sigma_B)\times (\Gamma_A \times
\Gamma_B)\right)^*\end{multline*}
(so $a_i\in \Sigma_A, a_i'\in \Sigma_B, b_i\in \Gamma_A$
and $b_i'\in \Gamma_B$) is \emph{consistent} with SDFST $\TT$ and strategies
$x_A, x_B$ if 
\begin{enumerate}[(i)]\item $w\in L(\TT)$, and
\item\label{it:strats} for every $1\leq i \leq k$ we have \[a_i =
x_A((a_1,b_1),\ldots, (a_{i-1}, b_{i-1})),\] and \[a_i' = x_B((a_1',b_1'),\ldots,
(a'_{i-1}, b'_{i-1})).\]
\end{enumerate}

A word $(a'_1,b'_1)\ldots (a'_k,b'_k)\in (\Sigma_B\times \Gamma_B)^*$ is
consistent with $\TT, x_A$ and $x_B$ if there exist $a_1,\ldots,a_k \in
\Sigma_A$ and $b_1,\ldots,b_k\in \Gamma_A$ such that
$(((a_1,a_1'),(b_1,b_1'))\ldots \allowbreak ((a_k,a_k'),(b_k,b_k')))$ is consistent with
$\TT, x_A$ and $x_B$.
\end{definition}

We will sometimes refer to limb (\ref{it:strats}) of the above Definition as
`being consistent with $x_A, x_B$'; then being consistent with $\TT, x_A, x_B$
means being an element of $L(\TT)$ and being consistent with $x_A,x_B$.

Could we choose to have $Y=(\Sigma_B\times \Gamma_B)^*$, and say that
$p_{Y|X_A,X_B}(y|x_A,x_B)=1$ if $y$ is consistent with $\TT, x_A$ and $x_B$?

No, because such a $y$ may not be unique, and so the matrix
$p_{Y|X_A,X_B}(y|x_A,x_B)$ would not in general be stochastic.  For example, if
$\TT$ is the identity transduction and $x_A$ and $x_B$ are both
constant $a$, we have that $(a,a)^k$ is consistent with $\TT, x_A$ and $x_B$ for
all $k$.  But prefixes are the only way this can happen.

\begin{proposition}\label{prop:uniqword}Let $\TT$ be an SDFST as above and let
$x_A,x_B$ be strategies for Alice and Bob.  Then there exists some $w_0\in
(\Sigma\times \Gamma)^{\omega}$ such that for any $w\in L(\TT)$ we have that
$w$ is consistent with $x_A$ and $x_B$ if and only if $w\leq w_0$.
\end{proposition}
\begin{proof}
Define the infinite word \[w_0 =
((a_1,a_1'),(b_1,b'_1))((a_2,a_2'),(b_2,b_2'))\ldots \in
(\Sigma\times\Gamma)^{\omega}\] by 
\begin{align*}a_i &= x_A((a_1,b_1)\ldots
(a_{i-1},b_{i-1})), \\
a_i' &= x_B((a'_1,b'_1)\ldots(a'_{i-1},b'_{i-1})), \text{ and} \\
(b_i,b'_i) &= \sigma(q_{i-1},(a_i,a'_i)),\end{align*}
where $q_0$ is the initial state and
the sequence $q_0q_1\ldots$ is defined by $q_i = \delta(q_{i-1},(a_i,a'_i))$ for
$i\geq 1$.

Clearly if $w \leq w_0$ then $w$ satisfies limb (\ref{it:strats}) of Definition
\ref{def:cons}, and so if also $w\in L(\TT)$ then $w$ is consistent with $\TT,
x_A$ and $x_B$.

Conversely suppose that $w\not\leq w_0$.  Then we have that $w = w'(a,b)w''$ for some
$w'=((a_1,a'_1),(b_1,b'_1)) \ldots \allowbreak ((a_k,a'_k),(b_k,b'_k)) \leq w_0$, some $w''\in
\Sigma\times\Gamma^*$ and some $(a,b)\in \Sigma\times \Gamma$ with $(a,b)\neq
((a_{k+1},a'_{k+1}),(b_{k+1},b'_{k+1}))$.  But if $a\neq (a_{k+1},a'_{k+1})$
then without loss of generality we have $\fst(a)\neq a_{k+1} =
x_A((a_1,b_1)\ldots(a_k,b_k))$ and so $w$ is not consistent with $x_A,x_B$.

On the other hand if $b\neq (b_{k+1},b'_{k+1}) = \sigma(q_k, (a_k,a'_k))$ then
$w \notin L(\TT)$.  Either way we have that $w$ is not consistent with $\TT,
x_A$ and $x_B$.
\end{proof}

The intuition here is that having fixed $x_A$ and $x_B$, these uniquely 
determine the actions of Alice and Bob at each step given the outputs they are 
shown, and $\TT$ determines those outputs uniquely based on the actions up to 
the current time.

Projecting $w_0$ onto $(\Sigma_B\times \Gamma_B)^{\omega}$ gives

\begin{corollary}\label{cor:uniqpref}
Let $\TT, x_A$ and $x_B$ be as above.  There exists some
$w_0\in (\Sigma_B\times \Gamma_B)^{\omega}$ such that if $w\in (\Sigma_B\times
\Gamma_B)^*$ is consistent with $\TT, x_A$ and $x_B$ then $w \leq w_0$.
\end{corollary}

So can we have $Y=(\Sigma_B\times \Gamma_B)^{\omega}$, and
$p_{Y|X_A,X_B}(y|x_A,x_B) = 1$ for $y=w_0$ as in Corollary \ref{cor:uniqpref}?

One reason why not is that this is not at all realistic: it corresponds to Bob
being able to conduct an experiment lasting for an infinite time.  Moreover it
would allow Bob to acquire an infinite (or at least unbounded) amount of
information, and it is not clear how this should be interpreted.

For this reason we will consider Bob's interaction with the system not as a
single experiment, but as a \emph{family} of experiments, parametrised by the
amount of time allowed; that is, by the length of traces which we consider as
outcomes.  Assuming for the moment that $F=Q$, we then have that the matrix
$p_{Y|X_A,X_B}$ is stochastic.

\begin{proposition}\label{prop:stochmat}
Let $\TT$ be an SDFST with $F=Q$, and let $Y=(\Sigma_B\times \Gamma_B)^k$ for
some $k\in \NN$.  Let the matrix $p_{Y|X_A,X_B}$ be defined by
$p_{Y|X_A,X_B}(y|x_A,x_B)=1$ if $y$ is compatible with $\TT, x_A$ and $x_B$, and
0 otherwise.  Then $p_{Y|X_A,X_B}$ is stochastic; that is, we have
\[\sum_{y\in Y} p_{Y|X_A,X_B}(y|x_A,x_B) = 1\] for all $x_A\in \XX_A$ and
$x_B\in \XX_B$.  
\end{proposition}
\begin{proof}
By Corollary \ref{cor:uniqpref}, we have that for fixed $x_A, x_B$ there is at
most one $y\in (\Sigma_B\times \Gamma_B)^k$ which is consistent with $\TT, x_A$
and $x_B$.  On the other hand it is clear from the definitions that if $F=Q$
then all prefixes of the infinite word $w_0$ from Proposition
\ref{prop:uniqword} are accepted by $\TT$.  Hence projecting $w_0$ onto
$(\Sigma_B\times \Gamma_B)^k$ gives a suitable $y$.
\end{proof}

Truncating at length $k$ also means that strategies $x_A, x_B$ can be viewed as
drawn from the spaces of functions $(\Sigma_A\times \Gamma_A)^{<k}\rightarrow
\Sigma_A$ and $(\Sigma_B\times \Gamma_B)^{<k}\rightarrow \Sigma_B$ respectively.
This means that the spaces $\XX_A$ and $\XX_B$ of possible strategies for Alice
and Bob are also finite.

We can now apply Theorem \ref{thm:aliceposs} to calculate the
information flow as the size of the largest possible set of outcomes that can
consistently be seen by Bob, and for convenience we will adopt this as a
definition.

\begin{definition}\label{def:intinf}
Let $\TT$ be an SDFST over input and output alphabets $\Sigma_A\times \Sigma_B$,
and let $\XX_A, \XX_B$ be the spaces of functions $(\Sigma_A\times
\Gamma_A)^*\rightarrow \Sigma_A$ and $(\Sigma_B\times \Gamma_B)^* \rightarrow
\Sigma_B$ respectively.  Define
\begin{multline*}\LL_k(\TT) = \max_{x_B\in \XX_B} \log\left| \left\{ y\in (\Sigma_B\times
\Gamma_B)^k \middle\vert \exists x_A\in \XX_A :\right.\right. \\ 
\left.\left. \text{$y$ is consistent with $\TT,
x_A$ and $x_B$}\right\} \right|.\end{multline*}
\end{definition}

Observe that if $F=Q$ then by Theorem \ref{thm:aliceposs} we have that
$\LL_k(\TT)=\LL_\infty(\C)$, where $\C$ is the interactive channel defined by
the matrix of conditional probabilities in the statement of Proposition
\ref{prop:stochmat}.

What about the case where $F \subsetneq Q$?  The treatment of this depends on
what we consider to be the meaning of a run ending in a non-accepting state.
One interpretation is that it represents a catastrophically bad outcome (say,
the intruder being detected) which must be avoided.  By Corollary
\ref{cor:bobdet} we may assume that Bob is employing a pure (i.e.
non-random) strategy, and so Alice can ensure that non-accepting runs are
avoided by avoiding particular $x_A$.  This means that Definition
\ref{def:intinf} is exactly right for this interpretation.

Another possible interpretation is that a run ending in a non-accepting state
produces some kind of `error' output, where all errors are indistinguishable.
This essentially increases the number of possible observations by Bob by either 1
or 0, depending on whether or not the extremal $x_B$ allows for non-accepting
runs.  This means that the amount of information is either $\LL_k(\TT)$ or $\log (1 +
2^{\LL_k(\TT)})$, which we consider to be a trivial difference.

A third possiblity of course is that we reject the very notion of a
non-accepting run, and consider only SDFSTs with $F=Q$.  Note that many kinds of
behaviour which may involve the system going into an `error' state and producing
only a fixed `dummy' output symbol can straightforwardly be modelled as an SDFST
with $F=Q$.

Which of these three options the reader considers most satisfactory is, to some
extent, a matter of personal taste.  However, since as noted above all are
modelled adequately by Definition \ref{def:intinf}, that is what we shall adopt
as the basic definition for the remainder of this analysis.

Definition \ref{def:intinf} is in some sense an intensional definition, in the
sense that it involves directly considering all possible strategies for Alice
and Bob.  It will be helpful to have a more extensional version. Definition
\ref{def:intinf} can be recast as \[\LL_k(\TT)=\max_{X\in\FF} \log|X|,\] where
$\FF\subseteq \PP((\Sigma_B\times \Gamma_B)^k)$ is the family of sets $X$ such
that there exists some $x_B\in\XX_B$ such that 
\begin{multline*}X=\left\{ y\in (\Sigma_B\times
\Gamma_B)^k \middle\vert \exists x_A\in \XX_A :\right.\\
\left. \text{$y$ is consistent with $\TT, x_A$ and $x_B$}\right\}.\end{multline*}
So having an extensional characterisation of $\LL_k(\TT)$
amounts to having a condition for a set $X$ to be a member of $\FF$.

\begin{theorem}\label{thm:firstdiff}
Let $\TT, \XX_A$ and $\XX_B$ be as above.
Let $\FF \subseteq \PP\left((\Sigma_B\times \Gamma_B)^*\right)$ be defined by
$Y\in \FF$ if and only if there exists some $x_B\in \XX_B$ such that 
\begin{multline*}Y = \left\{ y\in (\Sigma_B\times \Gamma_B)^* \middle\vert 
\exists x_A\in \XX_A :\right. \\
\left.\text{$y$ is consistent with $\TT, x_A$ and $x_B$}\right\}.\end{multline*}

Let $X\subseteq (\Sigma_B\times \Gamma_B)^*$ be arbitrary.  Then $X\subseteq X'$
for some $X'\in \FF$ if and only if 
\begin{enumerate}[(i)]
\item $X\subseteq \left. L(\TT) \right\rvert_{(\Sigma_B\times \Gamma_B)^*}$, and
\item $X$ does not contain two elements which first differ by an element of
$\Sigma_B$.  That is, we do not have $w_1,w_2\in X$ such that $w_1 = w (a,b) w'$
and $w_2 = w (a',b') w''$ with $w,w',w'' \in (\Sigma_B\times \Gamma)^*, a,a'\in
\Sigma_B$ and $b,b'\in \Gamma_B$ with $a\neq a'$,
\end{enumerate}
where the notation $\left. L(\TT)\right\rvert_{(\Sigma_B\times \Gamma_B)^*}$
means the projection of $L(\TT)\subseteq ((\Sigma_A\times \Sigma_B)\times
(\Gamma_A\times\Gamma_B))^*$ onto the set $(\Sigma_B\times\Gamma_B)^*$.
\end{theorem}
\begin{proof}
The `only if' direction is straightforward.  Part (i) is immediate from the
definitions, and for part (ii) we must have $a = x_B(w) = a'$ (for the relevant
$x_B$).

For the `if' direction, suppose that $X$ satisfies the two conditions in the
statement of the theorem.  Define the partial function $x:(\Sigma_B\times
\Gamma_B)^* \rightharpoondown \Sigma_B$ by $x(w') = a$ whenever $w'(a,b) \leq w$
for some $w\in X$ and some $b\in \Gamma_B$.  This is well-defined by condition
(ii).  Define $x_B:(\Sigma_B\times \Gamma_B)^*\rightarrow \Sigma_B$ to be $x$,
extended arbitrarily where $x$ is undefined.  We claim that 
\begin{multline*}X\subseteq Y = \left\{ y\in (\Sigma_B\times \Gamma_B)^* \middle\vert \exists
x_A\in \XX_A :\right. \\
\left. \text{$y$ is consistent with $\TT, x_A$ and $x_B$}\right\}.\end{multline*}

Indeed, let $w\in X$ be arbitrary.  Plainly $w$ is consistent with $x_B$.  Since $w\in
\left.L(\TT)\right\rvert_{(\Sigma_B\times\Gamma_B)^*}$, there exists some
$w'\in L(\TT)$ such that $\left. w'\right\rvert_{(\Sigma_B\times \Gamma_B)^*}
= w$.  Define the partial function $x':(\Sigma_A\times
\Gamma_A)^*\rightharpoondown \Sigma_A$ by $x'(w'') = a$ whenever $w'' (a,b) \leq
w'$ for some $b\in \Gamma_A$.  Let $x_A:(\Sigma_A\times \Gamma_A)^*\rightarrow
\Sigma_A$ be an arbitrary total extension of $x'$.  Then plainly $w'$ is
consistent with $x_A$, and is also consistent with $x_B$ since $w$ was.  Hence
$w$ is consistent with $\TT, x_A$ and $x_B$, as required.
\end{proof}

Truncating to length $k$, and observing that 
\[\max_{X\in \FF} \log|X| = \max_{X\subseteq X' \in \FF} \log|X|\]
gives

\begin{corollary}\label{cor:firstdiff}Let $\TT, \XX_A$ and $\XX_B$ be as above.  Then we have
\[\LL_k(\TT) = \max_{X\in \FF'_k} \log|X|,\]
where $\FF'_k\subseteq \PP\left(\left.L(\TT)_{=k}
\right\rvert_{(\Sigma_B\times\Gamma_B)^k}\right)$ is the collection of sets which do
not contain two words which first differ by an element of $\Sigma_B$ (and this
has the same meaning as in part (ii) of Theorem \ref{thm:firstdiff}).
\end{corollary}

\subsection{Reduction to automata}

In this section, we show how to reduce the problem of computing $\LL_k(\TT)$
from a problem about transducers to a problem which mentions only automata.  The
first step is to produce an automaton whose language is in correspondence with
Bob's interface with $\TT$.

\begin{definition}\label{def:transaut}Let $\TT = (Q, q_0, F, \Sigma_A\times \Sigma_B, \Gamma_A\times
\Gamma_B, \delta, \sigma)$ be an SDFST.  Define the nondeterministic finite
automaton $\A_{\TT} = (Q\cup(Q\times \Gamma_B), q_0, F, \Sigma_B\cup\Gamma_B,
\Delta)$, where
\[\Delta(q,a') = \left\{(\delta(q,(a,a')),\snd(\sigma(q,(a,a')))) \middle\vert a\in
\Sigma_A
\right\}\]
for all $q\in Q$ and $a'\in \Sigma_B$, $\Delta(q,b')=\emptyset$ for all $b'\in
\Gamma_B$, and
\[\Delta((q,b'),x) = \begin{cases}\{q\} &\quad \text{if $x=b'$} \\
\emptyset &\quad \text{otherwise}\end{cases}\]
for all $(q,b')\in Q\times \Gamma_B$ and $x\in \Sigma_B\cup \Gamma_B$.
\end{definition}

Informally, we introduce an auxiliary state for each pair $(q,b')\in Q\times
\Gamma_B$ to represent the behaviour `emit the event $b'$ and then go into state
$q$'.  For states $q,q'\in Q$ and events $a'\in \Sigma_B, b'\in \Gamma_B$ we
have a transition from $q$ to $(q',b')$ if and only there exist some $a\in
\Sigma_A$ and $b\in \Gamma_A$ such that $\delta(q,(a,a')) = q'$ and
$\sigma(q,(a,a')) = (b,b')$.  In the language of Communicating Sequential
Procceses, this corresponds to treating Alice's behavious as nondeterministic
and hiding all of her events: that is, the familiar \emph{lazy abstraction}
formulation of noninterference~\cite{roscoe1994non}.

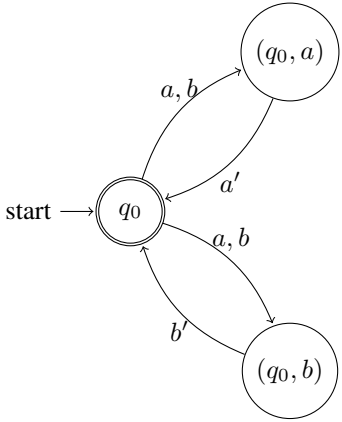
\begin{figure}[htbp]
\centering
\begin{tikzpicture}[shorten >=1pt,node distance=3cm,on grid,auto]
   \node[state,initial, accepting] (q_0)   {$q_0$};
   \node[state] (q_0a) [above right=of q_0] {$(q_0,a)$};
   \node[state] (q_0b) [below right=of q_0] {$(q_0,b)$};
   \path[->]
    (q_0) edge [bend left=25] node [above] {$a,b$} (q_0a)
          edge [bend left=25] node [above] {$a,b$} (q_0b)
    (q_0a) edge [bend left=25] node [below] {$a'$} (q_0)
    (q_0b) edge [bend left=25] node [below] {$b'$} (q_0);
\end{tikzpicture}
\caption{Automaton corresponding to the relay system transducer shown in Figure
\ref{fig:relay}.}\label{fig:transaut}
\end{figure}

The following lemma is immediate from the definitions, and expresses the fact
that the words accepted by $\A_\TT$ are in precise correspondence with the words
accepted by $\TT$, projected onto $\Sigma_B\times \Gamma_B$.

\begin{lemma}\label{lem:flat}
Let $f:(\Sigma_B\times \Gamma_B)^* \rightarrow (\Sigma_B\cup \Gamma_B)^*$ be the
\emph{flattening} operation defined by $f((a_1,b_1)\ldots(a_k,b_k)) =
a_1b_1\ldots a_kb_k$.  Then we have
\[L(\A_{\TT}) = f\left(\left. L(\TT)
\right\rvert_{(\Sigma_B\times\Gamma_B)^*}\right).\]
\end{lemma}

Note that since elements of $\Sigma_B$ appear at odd-numbered positions in
traces of $\A_\TT$ and elements of $\Gamma_B$ appear at even-numbered positions,
we may assume without loss of generality that $\Sigma_B$ and $\Gamma_B$ are
disjoint.  Then combining Lemma \ref{lem:flat} with Corollary
\ref{cor:firstdiff} gives

\begin{theorem}\label{thm:autdiff}Let $\TT$ be an SDFST as above such that $\Sigma_B$ and $\Gamma_B$
are disjoint.  Then
\[\LL_k(\TT) = \max_{X\in \FF_k} \log |X|,\]
where $\FF_k\subseteq \PP\left(L\left(\A_\TT\right)_{=2k} \right)$ is the
collections of sets which do not contain two words which first differ by an
element of $\Sigma_B$; that is, for $X\in \FF_{k}$ we do not have $w_1,w_2\in X$
with $w_1=w a w', w_2 = w a' w''$, with $w,w',w''\in (\Sigma_B\cup \Gamma_B)^*$
and $a\neq a' \in \Sigma_B$.
\end{theorem}

Note that an alternative notation for this theorem (and, \emph{mutatis
mutandis}, Corollary \ref{cor:firstdiff}) would be to define a single family
$\FF\subseteq \PP\left((\Sigma_B\cup \Gamma_B)^*\right)$ consisting of the sets
which do not contain words first differing on an element of $\Sigma_B$, and then
say that
\[\LL_k(\TT) = \max_{X\in \FF} \log\left|X\cap L(\A_\TT)_{=2k}\right|.\]

We have therefore reduced computing the information flow permitted by a
deterministic interactive system to an instance of a more general problem over
finite automata, which we call the \emph{$\Sigma$-deterministic subset growth}
problem.

\begin{definition}Let $\Sigma, \Gamma$ be disjoint finite sets.  A set
$X\subseteq (\Sigma\cup \Gamma)^*$ is
\emph{$\Sigma$-deterministic} if it does not contain two words which first
differ by an element of $\Sigma$; that is, we do not have $w_1,w_2\in X$ with
$w_1=waw', w_2=wa'w''$, with $w,w',w''\in (\Sigma\cup \Gamma)^*$ and $a\neq a'
\in \Sigma$.

For a nondeterministic finite automaton $\A$ over alphabet $\Sigma\cup \Gamma$,
define
\[D_k(\A) = \max_{X\in \FF_k} |X|,\]
where $\FF_k$ consists of the $\Sigma$-deterministic subsets of $L(\A)_{=k}$.
\end{definition}

\begin{problem}[$\Sigma$-deterministic subset growth]\label{prob:detgrowth}
Given a nondeterministic finite automaton $\A$ over $\Sigma\cup \Gamma$,
determine the growth rate of $D_k(\A)$.
\end{problem}

Of course, the statement of this problem is somewhat informal, in that
the meaning of `determine the growth rate' is not precisely specified.  This is
in some sense inevitable, considering that $D_k(\A)$ is an infinite collection
of values, so many types of results are possible.  Below we will obtain results 
on the asymptotic growth of $D_k(\A)$ as $k\rightarrow \infty$.

\subsection{Antichains}\label{sec:autantichains}

In this section we will see that Problem \ref{prob:detgrowth} can be further 
reduced, to that of computing the `width' of $L(\A)$.

\begin{definition}Let $X$ be a set, and let $\leq$ be a
partial order on $X$.  Then the \emph{lexicographic} order induced by $\leq$ on
$X^*$, denoted $\preceq$, is defined by 
\begin{align*}
\forall w\in X^*&: \epsilon \preceq w \text{ (and $w\not\preceq \epsilon$ if $w\neq
\epsilon$), and}\\
\forall x,y\in X, w,w' \in X^* &: xw \preceq yw' \text{ if and only if either
$x < y$} \\ &\qquad\qquad \text{or $x=y$ and $w\preceq w'$.}
\end{align*}
\end{definition}

Observe that $\preceq$ defines a partial order.  Indeed, suppose that $w_1,w_2\in
X^*$ are of minimum total length such that $w_1\preceq w_2, w_2\preceq w_1$ but
$w_1 \neq w_2$.  Trivially if $w_1 = \epsilon$ then also $w_2 = \epsilon = w_1$
(and vice versa).  Otherwise we have $w_1 = xw, w_2 = yw'$, and either $x < y$
or $x=y$ and $w\preceq w'$, and on the other hand either $y < x$ or $y=x$ and
$w'\preceq w$.  Hence we have $y=x$ and both $w\preceq w'$ and $w' \preceq w$, so by
induction $w=w'$.  Hence $w_1=w_2$, a contradiction, so indeed $\preceq$ is
antisymmetric.

Similarly suppose that $w_1, w_2, w_3\in X^*$ are of minimum total length such
that $w_1\preceq w_2$ and $w_2\preceq w_3$ but $w_1\not\preceq w_3$.  Since
$w_1\not\preceq w_3$ we have $w_1\neq \epsilon$, and plainly
$w_1\neq w_2$ and $w_2\neq w_3$ and so $w_2, w_3 \neq \epsilon$.  Write $w_1 =
xw, w_2 = yw'$ and $w_3 = zw''$.  If $x <y$ then (since $y\leq z$) we have $x <
z$ and so $w_1\preceq w_3$.  Similarly if $y<z$ then (since $x\leq y$) we have
$x<x$ so $w_1\preceq w_3$.  Hence we have $x=y=z$ and $w\preceq w'$ and $w' \preceq
w''$.  But then by induction we have $w\preceq w''$ and so $w_1\preceq w_3$, a
contradiction.  Hence indeed $\preceq$ is transitive and so (since we have also
shown it is antisymmetric, and it is trivially reflexive) it is a partial order.

\begin{figure}[htb]
\centering
\begin{tikzpicture}[node distance=1cm]
\node [] (info) {$\sup_{((X_A,X_B),Y)} \LL_\infty((X_A,X_B),Y)$};
\node [below=of info] (uniform) {$\begin{aligned}&\sup_{x_B\in \XX_B} \log \left| \left\{y\in\YY
\middle\vert \exists x_A\in \XX_A: \right.\right. \\ 
&\qquad\qquad\left. \left. p_{Y|X_A,X_B}(y|x_A,x_B)=1\right\}\right|\end{aligned}$};
\node [below=of uniform] (firstdiff) {$\begin{aligned}&\max \left\{\log |X|\middle\vert X\subseteq
\left.L(\TT)\right\rvert_{(\Sigma_B\times \Gamma_B)^k},\right. \\ &\qquad\qquad \left.\text{ no $w,w'\in X$
first differ in $\Sigma_B$}\right\} \end{aligned}$};
\node [below=of firstdiff] (auto) {$\max \left\{\log|X| \middle\vert X\subseteq
L\left(\A_\TT \right)_{=2k} \text{ is $\Sigma_B$-deterministic}\right\}$};
\node [below=of auto] (antichain) {$\max \left\{\log|X| \middle\vert X\subseteq
L\left(\A_\TT \right)_{=2k} \text{ is an antichain}\right\}$};

\path[->] (info) edge node [right] {Theorem \ref{thm:aliceposs}} (uniform)
  (uniform) edge node [right] {Corollary \ref{cor:firstdiff}} (firstdiff)
  (firstdiff) edge node [right] {Theorem \ref{thm:autdiff}} (auto)
  (auto) edge node [right] {Theorem \ref{thm:infanti}} (antichain); 
\end{tikzpicture}
\caption{The structure of Sections \ref{sec:channels} and \ref{sec:detint}}\label{fig:chapstruct}
\end{figure}
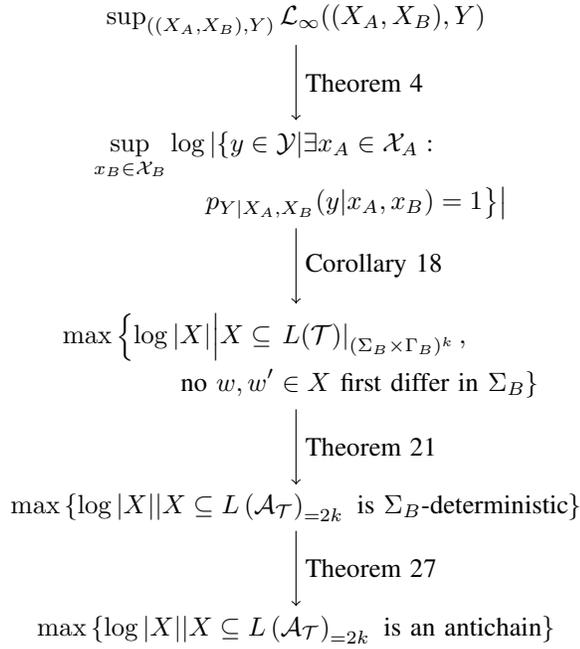

The study of partially ordered sets is often concerned with \emph{chains} (sets
wherein any two elements are comparable) and \emph{antichains} (sets where no
two elements are comparable).

\begin{definition}
Let $X$ be a partially ordered set, partially orderd by $\leq$.  
A set $Y\subseteq X$ is a \emph{chain} if
for any $x,y\in Y$ we have $x\leq y$ or $y\leq x$.  $Y$ is an \emph{antichain}
if for any $x,y \in Y$ such that $x\leq y$ we have $x=y$.  Let $Y \subseteq X$
be an antichain of maximum size.  Then $|Y|$ is the \emph{width} of $X$, denoted
$w(X)$.
\end{definition}

An example of the relevance of the width of a partially ordered set to its
structure is given by the celebrated theorem of Robert
Dilworth~\cite{dilworth1950decomposition}.

\begin{theorem}[Dilworth, 1950]
Let $X$ be a partially ordered set.  Let $k$ be minimal such that
$X=Y_1\cup\ldots \cup Y_k$ with each $Y_k$ a chain.  Then $k=w(X)$.
\end{theorem}

The relevance of these ideas to Problem \ref{prob:detgrowth} is established by
the following theorem.

\begin{theorem}\label{thm:infanti}Let $\Sigma, \Gamma$ be disjoint sets.  Define the partial order
$\leq$ on $\Sigma\cup \Gamma$ by setting $\left.\leq\right \rvert_{\Sigma}$ to
be an arbitrary linear order on $\Sigma$, and setting $x\not\leq y, y\not\leq x$ for all
$x\in \Gamma$ and all $y\in \Sigma\cup \Gamma$ with $y \neq x$.

Let $X\subseteq (\Sigma\cup\Gamma)^k$ be arbitrary.  Then $X$ is
$\Sigma$-deterministic if and only if it is an antichain with respect to the
lexicographic order induced by $\leq$.
\end{theorem}
\begin{proof}
If $w_1,w_2\in (\Sigma\cup \Gamma)^k$ first differ by an element of $\Sigma$,
say $w_1 = w a w'$, $w_2 = w a' w''$ with $a\neq a' \in \Sigma$.  Then without
loss of generality $a < a'$, so $w_1 \preceq w_2$.  Conversely, if $w_1 \preceq
w_2$, then write $w_1 = w x w', w_2 = w y w''$ for some $x\neq y \in \Sigma\cup
\Gamma$.  But then we must have $x < y$, and hence $x,y\in \Sigma$ and so
$w_1,w_2$ first differ by an element of $\Sigma$.
\end{proof}

We have thus reduced Problem \ref{prob:detgrowth} to the problem of calculating
the growth rate of the width of a regular language, with respect to this partial
order, a special case of the following problem.

\begin{problem}[Antichain growth for NFA]\label{prob:antgrowth}
Given a nondeterministic finite
automaton $\A$ over a finite partially ordered set $(\Sigma,\leq)$, determine
the growth rate of $w\left(L(\A)_{=k}\right)$, with respect to the
lexicographic order.
\end{problem}


The structure of the reductions in the preceding sections is shown in Figure
\ref{fig:chapstruct}.

Problem \ref{prob:antgrowth} is solved in \cite{mestel2018widths}, the
relevant results of which are summarised in Theorem \ref{thm:antgrowth}
(Theorems 16, 18, 25 and 28 of \cite{mestel2018widths}).


\begin{theorem}\label{thm:antgrowth}
Let $\A$ be an NFA over a partially ordered set $(\Sigma,\leq)$.  Then we have
the following:
\begin{enumerate}[(i)]
\item The antichain growth of $L(\A)$ is either polynomial or exponential.
That is, we have either $w(L(\A)_{=n}) = O(n^k)$ for some $k$ or
$w(L(\A)_{=n}) = \Omega(2^{\epsilon n})$ for some $\epsilon > 0$.
\item There is a polynomial-time algorithm to determine whether a given $\A$ has
polynomial or exponential antichain growth.
\item In the case of polynomial antichain growth, we have that
$w(L(\A)_{=n})=\Theta(n^k)$ for some integer $k$, and there is a
polynomial-time algorithm to compute $k$ for a given automaton.
\end{enumerate}
\end{theorem}

Combining Theorem \ref{thm:antgrowth} with the reduction shown in Figure
\ref{fig:chapstruct} yields the main theorem of this work, that any SDFST has
either logarithmic or linear min-entropy capacity, and there is a
polynomial-time algorithm to distinguish the two cases (and determine the
constant for logarithmic capacity).

\begin{theorem}\label{thm:entcap}
Let $\TT = (Q,q_0,F,\Sigma_A\times\Sigma_B,\Gamma_A\times\Gamma_B,\delta,
\sigma)$ be an SDFST.  Then we have the following:
\begin{enumerate}[(i)]
\item The min-entropy capacity $\LL_n(\TT)$ is either logarithmic or linear.
That is, we have either $\LL_n(\TT) = O(\log n)$ or $\LL_n(\TT) = \Theta(n)$.
\item There is a polynomial-time algorithm to determine whether a given $\TT$
has logarithmic or linear capacity growth.
\item In the case of logarithmic capacity, we have that $\LL_n(\TT) \sim
k\log n$ for some integer $k$, and there is a polynomial-time algorithm to
compute $k$ for a given SDFST.
\end{enumerate}
\end{theorem}

Note in particular that the information flow capacity is bounded if and only if 
$w(L(\A)_{=n})$ has polynomial growth of order 0.

Returning to the relay system shown in Figure \ref{fig:relay} at the beginning 
of this section, it is easy to see that the corresponding automaton shown in 
Figure \ref{fig:transaut} has exponential antichain growth, since in particular 
its language contains the exponential antichain $(aa'+ab')^*$.  We conclude that 
the system allows linear information flow, which is as expected since in $n$ 
steps Alice can transmit $n$ independent bits to Bob.

We claim that the cases of linear and logarithmic information flow can in some 
sense be interpreted as `dangerous' and `safe' respectively.  That linear 
information flow is dangerous should require no explanation: it offers an 
adversary an exponential speedup over exhaustive guessing of a secret (for instance 
a cryptographic key).  On the other hand, if the information flow in time $n$ is 
only proportional to $\log n$, then this offers the adversary at most a 
polynomial speedup over exhaustive guessing.

Of course it will not be appropriate in every situation to regard logarithmic 
antichain growth as `safe', and for instance we may sometimes be more interested in the 
precise amount of information flow that can occur in a fixed time $n$.  This is 
given by $w(L(\A_\TT)_{=n})$, which can be computed by a straightforward dynamic 
programming algorithm at the cost of determinising $\A_\TT$; see p.89 of the 
author's PhD thesis~\cite{mestelthesis} for details.  Whether there is 
an algorithm which is polynomial in $n$ and the size of $\A_\TT$ (as an NFA) is 
an open problem.


\subsection{Example: a simple scheduler}

We now illustrate the theory of the preceding two sections by applying it to 
analyse a simple scheduler. 
A resource is shared between Alice and Bob, and at each step Alice can transmit 
$a$, signifying that she wishes to use the resource, or $b$, signifying that she 
does not.  She receives back either an $a'$, signifying that she was succesful, 
or a $b'$, signifying that she was not (if she did not ask to use the resource 
then she always receives a $b'$).  The interface for Bob is similar but with 
primed and unprimed alphabets reversed.

Initially, Bob has priority over the use of the system, and for as long as
Alice transmits $b$ he retains it.  However, as soon as Alice seeks to use the
system by transmitting an $a$ she obtains priority and retains it for as long as
she uses it continuously.  As soon as she transmits a $b$ priority shifts back
to Bob, who retains it for the remainder of the execution.  

The transducer $\TT$ corresponding to this system is depicted in Figure
\ref{fig:interrupt} (where missing arguments mean that the input from that user
is ignored).

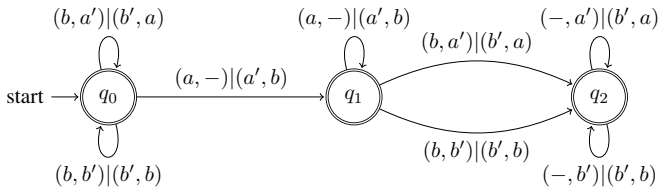
\begin{figure}[htb]
\centering
\resizebox{0.5\textwidth}{!}{
\begin{tikzpicture}[shorten >=1pt,node distance=4cm,on grid,auto]
   \node[state,initial, accepting] (q_0)   {$q_0$};
   \node[state,accepting] (q_1) [right=of q_0] {$q_1$};
   \node[state,accepting] (q_2) [right=of q_1] {$q_2$};
   \path[->]
    (q_0) edge [loop above] node {$(b,a')|(b',a)$} (q_0)
          edge [loop below] node {$(b,b')|(b',b)$} (q_0)
    (q_0) edge node [above] {$(a,-)|(a',b)$} (q_1)
    (q_1) edge [loop above] node {$(a,-)|(a',b)$} (q_1)
    (q_1) edge [bend left=25] node [above] {$(b,a')|(b',a)$} (q_2)
    (q_1) edge [bend right=25] node [below] {$(b,b')|(b',b)$} (q_2)
    (q_2) edge [loop above] node {$(-,a')|(b',a)$} (q_2)
          edge [loop below] node {$(-,b')|(b',b)$} (q_2);
\end{tikzpicture}
}
\caption{An interrupt system.}\label{fig:interrupt}
\end{figure}

We can now apply Definition \ref{def:transaut} to construct the corresponding
automaton $\A$, which is shown in Figure \ref{fig:interruptaut}.  By Theorems 
\ref{thm:autdiff} and \ref{thm:infanti} we have that $\LL_n(\TT) = w(L(\A)_{=k})$,
where $L(\A)$ is given the lexicographic order with the primed letters linearly 
ordered and the unprimed letters incomparable.

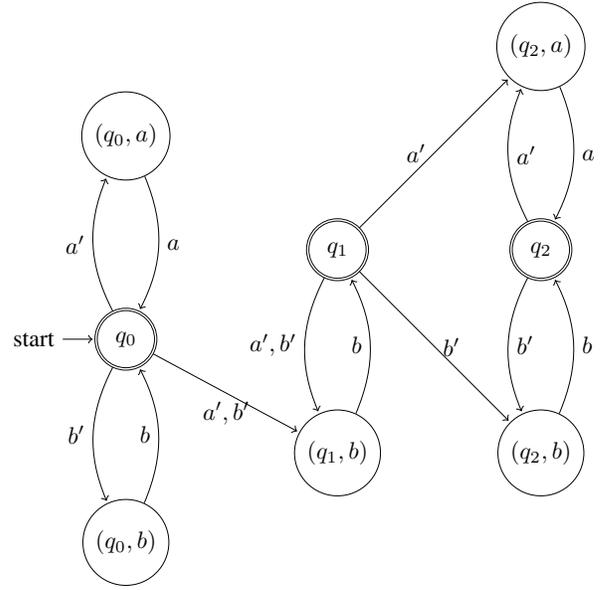
\begin{figure}[htb]
\centering
\resizebox{0.45\textwidth}{!}{
\begin{tikzpicture}[shorten >=1pt,node distance=3cm,on grid,auto]
   \node[state,initial, accepting] (q_0)   {$q_0$};
   \node[state] (q_0a) [above=of q_0] {$(q_0,a)$};
   \node[state] (q_0b) [below=of q_0] {$(q_0,b)$};
   \node[state, accepting] (q_1) [above right=of q_0,yshift=-0.8cm,xshift=1cm]   {$q_1$};
   \node[state] (q_1b) [below=of q_1] {$(q_1,b)$};
   \node[state, accepting] [right=of q_1] (q_2)   {$q_2$};
   \node[state] (q_2a) [above=of q_2] {$(q_2,a)$};
   \node[state] (q_2b) [below=of q_2] {$(q_2,b)$};

   \path[->]
    (q_0)  edge [bend left=25] node [left] {$a'$} (q_0a)
           edge [bend right=25] node [left] {$b'$} (q_0b)
    (q_0a) edge [bend left=25] node [right] {$a$} (q_0)
    (q_0b) edge [bend right=25] node [left] {$b$} (q_0)
    (q_0)  edge  node [below] {$a',b'$} (q_1b)
    (q_1b) edge [bend right=25] node [left] {$b$} (q_1)
    (q_1)  edge [bend right=25] node [left] {$a',b'$} (q_1b)
    (q_1)  edge node [left] {$a'$} (q_2a)
    (q_2)  edge [bend left=25] node [right] {$a'$} (q_2a)
    (q_2a) edge [bend left=25] node [right] {$a$} (q_2)
    (q_2)  edge [bend right=25] node [right] {$b'$} (q_2b)
    (q_2b) edge [bend right=25] node [right] {$b$} (q_2)
    (q_1)  edge node [right] {$b'$} (q_2b);
\end{tikzpicture}
}
\caption{Automaton corresponding to the interrupt system transducer shown in Figure
\ref{fig:interrupt}.}\label{fig:interruptaut}
\end{figure}

By the criteria in Theorems 16 and 28 of \cite{mestel2018widths}, this automaton 
has polynomial antichain growth of order 2, and so the system has logarithmic 
information flow, with $\LL_n(\TT)\sim 2\log n$ (see Section 6.5.2 of 
\cite{mestelthesis} for a more detailed discussion).  Note that this
makes intuitive sense: Alice can choose when to start using the resource and
when to stop, which she can do in $\binom{n}{2} = \Theta(n^2)$ ways. 

\section{Nondeterministic, multi-agent and probabilistic
systems}\label{sec:future}

In this section we describe some open problems relating to various
generalisations of the deterministic, two-agent systems considered in Section
\ref{sec:detint}.

\subsection{Nondeterministic systems}

In Section \ref{sec:detint} we considered only deterministic systems.  More
generally, however, we may be interested in systems which are nondeterministic:

\begin{definition} A \emph{synchronised nondeterministic finite-state
transducer} (SNDFST) is a 6-tuple $\TT=(Q,q_0,F,\Sigma,\Gamma,\Delta)$, where
$Q, q_0$ and $F$ are as in the definition of DFST, and $\Delta\subseteq Q\times
\Sigma \times Q \times \Gamma$ is the \emph{transition relation}.
\end{definition}

Similarly to before we say that $(a_1a_2\ldots a_k,b_1b_2\ldots b_k) \in
\Sigma^*\times \Gamma^*$ is accepted by $\TT$ if there exists a sequence of
states $q_1\ldots q_k\in Q^*$ such that $q_k\in F$ and for every $0 \leq i < k$
we have $(q_i, a_i, q_{i+1}, b_i) \in \Delta$.  As before we will consider
systems for which $\Sigma=\Sigma_A\times \Sigma_B$ and $\Gamma=\Gamma_A \times
\Gamma_B$, representing the inputs and outputs of Alice and Bob respectively.

The question then arises of how the nondeterminism in the system should be
interpreted.  One option is to consider it is essentially `demonic'---that is,
under the control of Alice and available to be used to convey information to
Bob.  This precisely corresponds to Definition \ref{def:intinf}, which can be
adopted wholesale, and a construction similar to that in Definition
\ref{def:transaut} can be used to produce an NFA $\A_\TT$ such that the capacity
of $\TT$ is equivalent to the antichain growth of $\A$.  We therefore have that
Theorem \ref{thm:entcap} holds also for nondeterministic systems interpreted in
this way.

However, the assumption of demonic nondeterminism may in some circumstances be
too pessimistic.  In particular, it may sometimes be reasonable to assume that
the way the nondeterminism is resolved depends only on the previous events, and
not on Alice's secret.  Equivalently, we may imagine that the resolution of the
nondeterminism is controlled by an `innocent' third party who is isolated from
both Alice and Bob (but is able to see their inputs and outputs).  We thus have
that if we can handle deterministic systems with multiple agents then we will be
able to handle nondeterministic systems with this interpretation.

\subsection{Multi-agent systems}

We will model multi-agent systems as SDSFTs, as before, but now with 
\begin{align*} \Sigma &=
\Sigma_A\times \Sigma_B \times \Sigma_1 \times \ldots \times \Sigma_k \text{
and}\\
\Gamma &= \Gamma_A \times \Gamma_B \times \Gamma_1 \times \ldots \times
\Gamma_k\end{align*}
for some $k$, where the $\Sigma_i$ and $\Gamma_i$ represent the inputs and
outputs respectively to the $i$th `innocent' agent.  We will call such a system
a $k$-SDSFT.

We will now require that the $k$ innocent agents choose distributions over
strategies.  An argument similar to Proposition \ref{prop:genpure} shows that we
may assume that the innocent agents select deterministic strategies, and so we
adopt a definition analagous to Definition \ref{def:intinf}.

\begin{definition}
Let $\TT$ be a $k$-SDFST.  We define
\begin{multline*}
\LL_n(\TT) = \max_{\substack{x_B\in \XX_B, \\ x_1\in \XX_1,\ldots,x_k \in \XX_k}} \log \left|
\left\{ y\in (\Sigma_B\times \Gamma_B)^k \middle\vert \exists x_A \in \XX_A:
\right.\right. \\ \left.\left. \text{$y$ is consistent with $\TT,x_A,x_B,x_1,\ldots,x_k$} \right\}
\right|,
\end{multline*}
where $\XX_i$ is the set of functions $(\Sigma_i\times \Gamma_i) \rightarrow
\Sigma_i$, and consistency is defined similarly to Definition \ref{def:cons}.
\end{definition}

Our first open problem is to compute the min-entropy capacity of a $k$-SDFST.
We conjecture that there should still be a dichotomy between polynomial and
exponential growth.

\subsection{Probabilistic systems}

We may also wish to handle systems whose behaviour is probabilistic.  We model
such systems as probabilistic finite-state transducers.

\begin{definition}
A \emph{probabilistic finite-state transducer} is a tuple $\TT=(Q,q_0,\Sigma,
\Gamma,\Delta)$, where $Q$ is a finite set of \emph{states}, $q_0\in Q$ is the
\emph{initial state} and $\Delta:Q\times \Sigma\times Q \times \Gamma
\rightarrow \mathbb{R}^{\geq 0}$ is the \emph{transition function}, such that
for all $q\in Q$ and all $a\in \Sigma$ we have
\[\sum_{b\in \Gamma, q'\in Q} \Delta(q,a,q',b) = 1.\]
\end{definition}

We interpret $\Delta(q,a,q',b)$ as the probability that on receiving the input
$a$ in state $q$, the system outputs $b$ and moves to state $q'$.

As before we require that $\Sigma$ and $\Gamma$ are of the form $\Sigma_A\times
\Sigma_B$ and $\Gamma_A\times \Gamma_B$ respectively (although of course it
would also be possible to consider multi-agent probabilistic systems), and the
sets $\XX_A$ and $\XX_B$ are as before.  For fixed $x_A$ and $x_B$, the output
$Y$ produced to Bob after $n$ steps is a sequence $y\in (\Sigma_B\times
\Gamma_B)^n$, where $y=((a'_1,b'_1),\ldots,(a'_n,b'_n))$ occurs with probability
\[\sum_{((a_1,b_1),\ldots, (a_n,b_n))\in Z} \sum_{q_1,\ldots,q_n \in Q}
\prod_{i=1}^n \Delta(q_{i-1},(a_i,a'_i),q_i,(b_i,b'_i)),\]
where $Z$ is the set of $((a_1,b_1),\ldots,(a_n,b_n)) \in (\Sigma_A\times
\Gamma_A)^n$ such that $((a_1,a'_1),(b_1,b'_1))\ldots \allowbreak
((a_n,a'_n),(b_n,b'_n))$ is consistent with $x_A$ and $x_B$.

This defines an interactive channel $\C_n$, and so our second open problem is to
compute the growth of $\LL_\infty(\C_n)$.  This seems to be a rather formidable
challenge since we lack a way to reduce to a possibilistic view of Alice's
actions, and so we may genuinely have to quantify over probability distributions for
$X_A$.

\section{Related work}\label{sec:relatedwork}

So far as we are aware this is the first quantitative study which is able to
analyse interactive systems in full generality, that is to say where inputs may
be provided by both parties, according to distributions chosen adversarially so
as to maximise information flow, rather than being specified as part of the
system.

Mardziel, Alvim, Hicks and Clarkson in \cite{mardziel2014quantifying} consider
interactive systems in essentially the same model as we use in Section
\ref{sec:detint} of this paper: they represent the system by a
probabilistic finite automaton, which is executed in a `context' consisting of
the strategy functions for the high and low users. They then employ
probabilistic programming to analyse particular systems with respect to
particular contexts, demonstrating for instance that allowing an adaptive
adversary can greatly increase information flow. However, they acknowledge that
they are not able to analyse the maximum leakage over all possible contexts,
instead observing that `We consider such worst-case reasoning challenging future
work'. The present work addresses this question for the case where the system is
deterministic.

K\"{o}pf and Basin in \cite{kopf2007infomodel} show how to calculate information
leakage for a particular model relating to side-channel attacks in which the
attacker is repeatedly permitted to make queries drawn from some fixed set. They
give an exhaustive algorithm to compute the maximum amount of information
leakege after $n$ queries.  

Boreale and Pampaloni in \cite{boreale2015quantitative} consider the case of 
repeated queries issued by the attacker (possibly adaptively) to
a stateless system and show that under certain reasonable assumptions the problem
of computing the maximum leakage after $n$ queries is NP-hard. In
\cite{boreale2011asymptotic}, the same authors together with Paolini study the
asymptotics of the leakage resulting from $n$ independent uses of a single
channel for large $n$. This is in some sense dual to the situation we have 
considered, of the asymptotics of a single, long execution of a stateful system.

In \cite{andres2010computing}, Andr{\'e}s, Palamidessi, van Rossum and Smith
compute the leakage of what they term `interactive information-hiding systems'
(IIHS), which are essentially automata over (secret) inputs and (observed)
outputs. However, they assume an essentially passive attacker: apart from the
values of the secret (whose distribution they sometimes allow to be chosen so as
to maximise information flow), the system is assumed to follow known
probabilistic behaviour. In follow-up work \cite{alvim2012quantitative}, Alvim,
Andr{\'e}s and Palamidessi demonstrate interesting connections between the
mutual information capacity of such systems and the directed information
capacity of channels with feedback, although this is of limited practical
significance since it is now recognised that mutual information is not generally an
appropriate measure of information flow.

An interesting alternative algorithmic approach is taken by Kawamoto and
Given-Wilson in \cite{kawamoto2015quantitative}, although for a completely
different problem from that addressed in this work. In
\cite{kawamoto2015quantitative}, the authors consider a purely passive
observer who is shown the outputs of two channels, interleaved according to some
scheduler; the goal is to find a scheduler which minimises the information
leakage. They show that this can be expressed as a linear programming
problem, and therefore solved in time polynomial in the number of possible
interleavings, which unfortunately is exponential in the number of possible
traces.

\section{Conclusions}\label{sec:conclusion}

In \cite{ryan2001noninterference}, Ryan, McLean, Millen and Gligor write the 
following:
\begin{quote}
Even at a theoretical level where timings are not available, and a bit per 
millisecond is not distinguishable from a bit per fortnight or a bit per century, 
a channel that compromises an unbounded amount of information is substantially 
different from one that cannot.  Characterization of unbounded channels is 
suggested as the kind of goal that would advance the study of this subject
\end{quote}
In Theorem \ref{thm:entcap} we have achieved this goal for deterministic 
systems, and in fact slightly more: we have shown that even among unbounded 
channels there is a dichotomy between `safe' and `dangerous' information flow, 
and this can be determined for a given system in polynomial time.

Having characterised the notion of safe versus dangerous information flow, one 
may ask about the question of enforcement of the safety criterion.  In one sense 
this question is already answered by Theorem \ref{thm:entcap}, since it includes 
a polynomial-time algorithm to determine whether the condition is satisfied for 
a given system.  However, the development of automated tools implementing this 
algorithm, which preferably would allow realistic systems to be specified using more convenient notation 
than the rather abstract mathematical formalism of finite-state transducers, is 
certainly an important area for future work.

\subsection*{Acknowledgements}

The author is grateful to Catuscia Palamidessi for helpful comments on an 
earlier version of this work, and to Dimiter Ostrev for comments on the final 
version.

\bibliographystyle{ieeetran}
\bibliography{mybib}

\end{document}